\documentclass[a4paper,12pt]{article}
\usepackage[T1]{fontenc}
\usepackage[english]{babel}
\usepackage[latin1]{inputenc}
\usepackage{amsfonts}
\usepackage{amsmath} 
\usepackage{amssymb}
\usepackage{amsthm}
\usepackage{appendix}
\usepackage{caption}
\usepackage{cases}

\usepackage{manfnt}
\usepackage{mathrsfs}
\usepackage{minitoc}
\usepackage[numbers]{natbib}

\usepackage{pgf}
\usepackage{pifont}
\usepackage{pstricks}

\usepackage{graphics}
\usepackage{graphicx}

\usepackage{lastpage}
\usepackage{latexsym}

\usepackage{color}
\usepackage{comment}

\usepackage{empheq}
\usepackage{enumerate}

\usepackage{fancyhdr}
\usepackage[top=2.5cm, bottom=3cm, left=2cm , right=2.1cm]{geometry}%Ajustement des marges
\theoremstyle{plain}
\newtheorem{Theorem}{Theorem}[section]

\newtheorem{Lemma}[Theorem]{Lemma}
\newtheorem{Definition}[Theorem]{Definition}
\newtheorem{Proposition}[Theorem]{Proposition}
\newtheorem{Remark}{Remark}[section]

\newtheorem{Assumption}[Remark]{Assumption}

\numberwithin{equation}{section}

\def\ind{{\bf 1}}
\def\demi{\frac{1}{2}}

\def\cal{\mathcal}

\def\B{{\mathcal B}}
\def\C{{\mathcal C}}

\def\F{{\mathcal F}}

\def\sigR{{\cal R}}
\def\sigoR{{\cal R}^{\perp}}
\def\sR{\text{{\tiny${\cal R}$}}}

\def\X{{\cal X}}

\def\Y{{\cal Y}}

\def\bfU{{\bf U}}

\def\E{{\mathbb E}}

\def\Q{{\mathbb Q}}
\def\R{{\mathbb R}}

\def\R{{\mathbb R}}

\def\tU{{\widetilde U}}
\def\tu{{\tilde u}}
\def\tv{{\tilde v}}

\def\wL^*{{\widehat L^{(\mu^*, \sigma^*)}}}
\def\wL{{\widehat L}}

\def\tV{{\widetilde V}}

\def\GX{{\mathscr X}}
\def\GY{{\mathscr Y}}

\newcommand{\rmi}{{\rm (i) $\>\>$}}

\newcommand{\rmii}{{\rm (ii) $\hspace{1.5mm}$}}
\newcommand{\rmiii}{{\rm (iii)$\>\>$}}
\newcommand{\rmiv}{{\rm (iv)$\>\>$}}

\newcommand{\rma}{{\rm a)$\>\>$}}
\newcommand{\rmb}{{\rm b)$\>\>$}}
\newcommand{\rmc}{{\rm c)$\>\>$}}

\def\B{\Big}
\def\b{\big}

\def\bit{\begin{itemize}}
\def\eit{\end{itemize}}

\def\bc{\begin{center}}
\def\ec{\end{center}}

\def\super { \end{document}}
\def\bcom{\begin{comment}}
\def\ecom{\end{comment}}
\def\edoc{\end{document}}

\DeclareMathOperator{\esssup}{ess\,sup}

\title{Ramsey Rule with Progressive  utility\\
 and Long Term Affine Yields Curves}

\author{ El Karoui Nicole,
\thanks{ \small  LPMA, UMR CNRS  6632,  Universit\'e Pierre et Marie Curie, CMAP, UMR CNRS 7641, \'Ecole Polytechnique  }
\\
\and Mrad~Mohamed~ \thanks
{\small   LAGA, UMR CNRS 7539,  Universit\'e Paris 13}
\and Hillairet Caroline~ \thanks
{\small CMAP, UMR CNRS 7641, \'Ecole Polytechnique,}
\thanks{With the financial
support of the "Chaire Risque Financier of the  Fondation du Risque", } 
}
\date{February 21, 2014}

\begin{document}
 \maketitle
 %$$$$$$$$$$$$$$$$$$$$$$$$$$$$$$$$
 \abstract{The  purpose of this paper relies on  the study of long term affine yield curves modeling.  It is inspired by the
Ramsey rule of the economic literature,  that links discount rate and marginal utility of aggregate optimal consumption. For
such a  long maturity modelization, the possibility of adjusting
preferences to new economic information is crucial, justifying the use of progressive utility. This paper studies, in a framework with affine factors, the yield curve given from the Ramsey rule. 
It first characterizes consistent progressive utility of investment and consumption, given the optimal wealth and consumption processes. A special attention is paid to utilities associated with linear optimal processes with respect to their initial conditions, which  is for example the case of power progressive utilities. Those utilities are the basis point to construct other progressive utilities generating non linear optimal processes but leading yet to still tractable computations.  This is of particular interest to study the impact of initial wealth on yield curves.  }\\


\section*{Introduction}

This paper focuses on the modelization of long term affine yield curves. For the financing of ecological project, for the pricing of longevity-linked securities or any other investment with long term impact, modeling long term interest rates is crucial. The answer cannot be find in financial market since  for longer maturities, the bond market is highly illiquid and standard financial interest rates models cannot be easily extended. Nevertheless, an abundant literature
on the economic aspects of long-term policy-making has been developed.  
The Ramsey rule, introduced by Ramsey in his seminal work \cite{Ramsey} and further discussed by numerous economists such as Gollier \cite{Gollier3,Gollier6,Gollier9,Gollier13,Gollier14,Gollier15,Gollier16,GollierEcological} and Weitzman \cite{Weitzman,
Weitzman_review}, is the reference equation to compute discount rate, that allows to evaluate the future value of an investment by giving a current equivalent value. The Ramsey rule links the discount rate with the marginal utility of aggregate consumption at the economic equilibrium. Even if this rule is very simple, there is no consensus among economists about the parameters that should be considered, leading to very different discount rates.  But economists agree on the necessity of a sequential decision scheme that allows to revise the first decisions in the light of new  knowledge  and  direct experiences:  the utility criterion must be adaptative and adjusted to the information flow. In the classical optimization point of view, this adaptative criteria is called  consistency. 
In that sense, market-consistent progressive utilities,  studied in El Karoui and Mrad
\cite{MrNek01,MrNek02,MrNek03},  are the appropriate tools to study long term yield curves.

 Indeed, in  a dynamic and stochastic environment, the classical notion of utility function  is not
flexible enough  to help us to make good choices in the long run. M. Musiela and T. Zariphopoulou (2003-2008
\cite{zar-03,zar-07a,zar-08,zar-07})  were
the first to suggest to use instead of the classical criterion  the concept of
 progressive dynamic utility, consistent with respect to a given investment universe in a sense specified in Section 1.
The concept of progressive utility  gives an adaptative way to  model possible changes over the time of individual preferences of an agent. 
In continuation of the recent works of El Karoui and Mrad \cite{MrNek01, MrNek02, MrNek03}, and motived by the Ramsey rule (in which the consumption rate is a key process), \cite{MrNek04} extends the notion of market-consistent progressive  utility to the case with consumption: the agent invest in a financial market and consumes a part of her wealth at each instant. As an example, backward classical value function is a progressive utility, the way the classical optimization problem is
posed is very different from the progressive utility problem. In the classical approach,  the optimal  processes are   computed through a backward analysis, 
emphasizing their dependency to the horizon  of the optimization problem, while the forward point of view makes clear the monotony of the optimal processes to their initial conditions. A special attention is paid to progressive utilities generating  linear optimal processes with respect to their initial conditions, which  is for example the case of power progressive utilities.

As the zero-coupon bond market is highly illiquid for long maturity,  it is relevant, for small trades,  to  give utility indifference price (also called Davis price) for zero coupon, using   progressive utility with consumption.  We study then the dynamics of the marginal utility yield curve, in the framework of
progressive and backward power utilities (since power utilities are the most commonly used in the economic literature)
and in a model with affine factors, since this model has the advantage to lead to tractable computations while allowing for more stochasticity than the log normal model studied in \cite{MrNek04}. Nevertheless, 
using power utilities implies that  the impact of the  initial economic wealth is avoided, since  in this case the optimal processes are linear with respect to the initial conditions. We thus propose a way of constructing, from power utilities,  progressive utilities generating non linear optimal processes but leading yet to still tractable computations.   The impact of the initial wealth   for yield curves is discussed.

The paper is organized as follows. After introducing the investment universe, 
Section 1 characterizes consistent progressive utility of investment and consumption, given the optimal wealth and consumption processes. Section 2 deals  with the computation of the marginal utility yield curve, inspired by the Ramsey rule.
Section 3 focuses on the yield curve with affine factors, in such a setting the yield curve does not depend on the initial
wealth of the economy. Section 4 provides then a modelization for yield curves dynamics that are non-linear to initial
conditions.

\section{ Progressive Utility and Investment Universe}

\subsection{The investment universe}\label{descmarket}
We consider an incomplete It\^o market, equipped with a $n$-standard Brownian motion,
$W$  and characterized by an adapted short rate $(r_t)$ and an adapted
$n$-dimensional risk premium vector 
$(\eta_t)$. All these processes are defined on a filtered probability space $(\Omega,{\cal F}_{t},\mathbb{P})$
satisfying usual assumptions; they are progressively processes satisfying minimal integrability assumptions, as $\int_0^T(r_t+\|\eta_t\|^2 dt)<\infty, a.s.$.\\
The agent  may invest in this  financial market and   is allowed to consume a part of his wealth at
the rate $C_t\geq 0$.
 %assumed to be an adapted and positive process. 
 To be short, we  give the mathematical definition of  the
class of admissible strategies $(\kappa_t, C_t)$, without specifying the risky assets. Nevertheless, the incompleteness of
the market is expressed by restrictions on the risky strategies constrained to live in a given progressive vector space
$(\sigR_t)$,  often obtained as the range of some progressive linear operator $\Im_t$.

\begin{Definition}[Test processes] \label{TestP} \rmi The self-financing dynamics of a  wealth process with risky portfolio $\kappa$ and consumption rate $C$  is given by
\begin{equation}\label{eq:DynamX}
dX^{\kappa,C}_t= X^{\kappa,C}_t[r_tdt +\kappa_t(dW_t+\eta_tdt)] -C_t\,dt, ~ \quad
\kappa_t \in \sigR_t.
\end{equation}
where $C$ is a positive progressive process,  $\kappa$ is a  progressive $n$-dimensional vector in $\sigR_t$, such that $\int_0^TC_t+\|\kappa_t\|^2 dt<\infty, a.s.$.\\ %measuring the volatility vector of the wealth $X^{\kappa,c}$.\\
\rmii A strategy $(\kappa_t,C_t)$ is said to be admissible if it is stopped with  the bankruptcy of the investor (when the wealth process
reaches $0$). \\%and if  the  portfolio $\kappa$ satisfies the following  restriction: we assume there exists progressive family of vector spaces $(\sigR_t)$  such that for any $t$, 
\rmiii The set of the wealth processes with admissible $(\kappa_t,C_t)$, also called test processes, is denoted by $\GX^c$. {When  portfolios are starting from $x$ at time $t$, we use the notation $\GX^c_t(x).$}
 \end{Definition}

\noindent The following short notations will be used extensively. 
Let $\sigR$ be a vector subspace of  $\R^n$. For any $x\in \R^n$,
$x^\sR$ is the orthogonal projection of the vector $x$
onto $\sigR$ and $x^{\perp}$ is the orthogonal projection
onto $\sigoR$. \\
The existence of  a risk premium  $\eta$ is a possible
formulation of the absence of arbitrage opportunity. From
Equation \eqref{eq:DynamX}, the minimal state price process $Y^0_t$, whose the dynamics is $dY^0_t=Y^0_t[-r_tdt+ (\nu_t-\eta^\sR_t).dW_t]$, belongs to the convex family $\GY$ of  positive It\^o's processes $Y_t$ such that $(Y_t X_t^{\kappa,C} + \int_0^t Y_s C_s ds) $ is a local martingale for any admissible portfolio. The existence of equivalent martingale measure is obtained by the assumption that the exponential local martingale
$L^{\eta^\sR}_t=\exp(-\int_0^t \eta^\sR_s.dW_s-\frac{1}{2} \int_0^t
||\eta^\sR_s||^2\,ds)$ is a uniformly integrable martingale. Nevertheless, we are interested into the class of
the so-called state price processes $Y_t$  belonging to the family $\GY$ characterized below.
%%�������������������������������
\begin{comment}
Since from
\eqref{eq:DynamX}, the impact of the risk premium on the wealth dynamics only
appears through the term $\kappa_t. \eta_t $ for  $\kappa_t\in \sigR_t$, there is
a "minimal" risk premium $(\eta^\sR_t)$, the projection of $\eta_t$  on
the space $\sigR_t$ $(\kappa_t. \eta_t =\kappa_t. \eta^\sR_t )$,
to which we refer in the sequel. The minimal state price density $Y^0_t$ whose the dynamics is $dY^\nu_t=Y^\nu_t[-r_tdt+ (\nu_t-\eta^\sR_t).dW_t]$ belongs to the convex family $\GY$ of  positive processes $Y_y$ such that $(Y_t X_.^{\kappa,c} + \int_0^t Y^_s c_s ds) $ is a local martingale.

Moreover, the existence of $\eta^\sR$ is not
enough to insure the existence of equivalent martingale measure, since in
general we do not know if the exponential local martingale
$L^{\eta^\sR}_t=\exp(-\int_0^t \eta^\sR_s.dW_s-\frac{1}{2} \int_0^t
||\eta^\sR_s||^2\,ds)$ is a uniformly integrable martingale, density of an
equivalent martingale measure. 
(taking into account the
discount factor) who will play the same role for the  conjugate utility,
than the test processes $X^{\kappa,c}$ for the  utility.
%%
\end{comment}
%%�������������������������������

\begin{Definition}[State price process] \label{SPDP} \rmi An It\^o semimartingale
$Y_t$ is called a state price process in $\GY$ if  for any test process
$X^{\kappa,C},~~\kappa \in \sigR$, \\
\centerline{$(Y_t X_t^{\kappa,C} + \int_0^t Y_s C_s ds) $ is a local martingale. }\\[2mm]
\rmii  This property is equivalent to the existence of progressive process $\nu_t\in \sigoR_t, (\int_0^T\|\nu_t\|^2 dt<\infty, a.s.)$ 
such that $Y=Y^\nu$ where $Y^\nu$ is the product of $Y^0$ $(\nu=0)$ by the exponential local martingale
$L^{\nu}_t=\exp\big(\int_0^t \nu_s.dW_s-1/2\int_0^t ||\nu_s||^2ds\big)$, and satisfies
\begin{equation}\label{Ynu}
dY^\nu_t=Y^\nu_t[-r_tdt+ (\nu_t-\eta^\sR_t).dW_t],\quad \nu_t\in \sigoR_t
 \quad Y^\nu_0=y
\end{equation}
\end{Definition}
\noindent From now on, to stress out the dependency on the initial condition, the solution of \eqref{Ynu} with initial condition $y$ will be denoted $(Y^\nu_t(y))$ and $Y^\nu_t:=Y^\nu_t(1) $;  the solution of \eqref{eq:DynamX} with initial condition $x$ will be denoted $(X^{\kappa,C}_t(x))$ and $X^{\kappa,C}_t:=X^{\kappa,C}_t(1)$.
%=======================================================

\subsection{$\GX^c$-consistent Utility and Portfolio optimization with consumption}
In long term (wealth-consumption) optimization problems, it is useful to have the choice to adapt utility criteria to  deep macro-evolution of economic environment. The concept  of progressive utility is introduced in this sense. As we are interested  in optimizing both the terminal wealth and the consumption rate, we introduce two progressive utilities $\bf(U,V)$, $\bf U$ for  the terminal wealth and $\bf V$ for   the consumption rate, often called utility system. 
%thses notionThe utility $\bf U$ from  the terminal wealth  is assumed to be It\^o's progressive utility; while  the utility  $\bf V$ from  the consumption rate is just assumed to be a progressive utility. (remark that from a dynamic point of view,
% $\bf U$ and $\bf V$  will play different roles). 
For sake of completeness, we start refer the reader to \cite{MrNek01} for a detailled study.

\begin{Definition}[{Progressive Utility}]\label{defPUF}$ $ \\
 \rmi A {\em progressive utility} is a ${\mathcal C}^2$- progressive random field on $\R^*_+$,
$\bfU=\{U(t, x); t \geq 0, x > 0\}$, starting from the deterministic utility function $u$ at time $0$, such that for every $(t,\omega)$, $x\mapsto U(\omega, t, x)$
is a strictly concave, strictly increasing, and 
non negative utility function, and satisfying thee Inada conditions:

$-$ for every $(t,\omega)$, $U(t,\omega,x)$ goes to $0$ when $x$ goes to $0$ 

 $-$  the derivative $U_x(t,\omega,x)$ (also called marginal utility) goes to $\infty$ when $x$ goes to $0$, 

 $-$  the derivative $U_x(t,\omega,x)$ goes to $0$ when $x$ goes to $\infty$.

\end{Definition}
\noindent For $t=0$, the deterministic utilities $U(0,.)$ and $V(0,.)$ are denoted $u(.)$ and $v(.)$ and in the following small letters $u$ and $v$ design deterministic utilities while capital letters refer to progressive utilities.  

As in statistical learning, the utility criterium is dynamically
adjusted to be the best given the market past information. So, market inputs may be
viewed as a calibration universe through the  test-class $\GX^c$ of processes on which the utility is
chosen to provide the best satisfaction. This motivates the following  definition of $\GX^c$-consistent utility system.
\begin{Definition}\label{def:conso}
A $\GX^c$-consistent progressive  utility  system of investment and consumption
is a pair of progressive utilities $\bf U$ and $\bf V$ on $\Omega\times[0,+\infty)\times\R^+$ such that,\\
%with the following additional properties:\\
\rmi {\sc Consistency with the test-class:}  For
any admissible wealth process  $X^{\kappa,C}\in \GX^c$, 
\begin{equation*}
 \mathbb{E}\big(U(t,X^{\kappa,C}_t)+\int_s^tV(s,C_s)ds /{\cal F}_s\big)\leq
U(s,X^{\kappa,C}_s), \>~ \forall s\leq t~~ 
a.s. 
\end{equation*}
In other words, the process $\b(U(t,X^{\kappa,C}_t)+\int_0^tV(s,C_s)ds\b)$ is a positive supermartingale, stopped at the first time of bankruptcy.\\
\rmii {\sc Existence of optimal strategy:} 
For any initial wealth $x>0$, there exists an optimal strategy $(\kappa^*,C^*)$ such that the associated non negative
wealth process $X^{*}=X^{\kappa^*,C^*}\in \GX^c$ issued from $x$ satisfies  $\big(U(t,X^*_t)+\int_0^tV(s,C^*_s)ds\big)$ is a
local martingale.\\
\rmiii To summarize, $U(t,x)$ is the value function of optimization problem with optimal strategies, that is for any maturity $T\geq t$
\begin{equation}\label{myopic value function}
U(t,x)=\esssup_{X^{\kappa,C}\in \GX^c_t(x)} \mathbb{E}\big(U(T,X^{\kappa,C}_T)+\int_t^T V(s,C_s)\ind_{\{X^{\kappa,C}_s(x)\geq 0\}}ds |{\cal F}_t\big) a.s.
\end{equation}
{ The optimal strategy $(X^{*},C^*)$ which is optimal for all these problems, independently of the time-horizon $T$, is called a myopic strategy.}\\
\rmiv  {\sc Strongly $\GX^c$-consistency} 
The system $(\bf U,\bf V)$ is said to be strongly $\GX^c$-consistent if the optimal process $X^{*}(x)$ is strictly increasing with respect to the initial condition $x$.
 \end{Definition}
\noindent Convex analysis  showed  the interest to introduce the convex conjugate utilities $\tilde{U}$  and  $\tilde{V}$  
defined as the Fenchel-Legendre
random field $\tilde U(t,y)= \sup_{c \geq 0, c \in \mathcal{Q}^+}(U(t,c)-cy)$ (similarly for $\tilde{V}$). Under mild regularity assumption, we have the following results (Karatzas-Shreve \cite{KaratzasShreve:01}, Rogers \cite{Rogers}).
\begin{Proposition}[Duality]\label{Duale}
Let $(U,V)$ be a pair of  stochastic $\GX^c$-consistent utilities with optimal strategy $(\kappa^*,C^*)$ leading to the non negative
wealth process $X^{*}=X^{\kappa^*,C^*}$. Then the convex
conjugate system $(\tilde{U},\tilde{V})$  satisfies :\\
\rmi For any admissible state price density process $Y^\nu\in \GY$ with $\nu\in \sigoR$,
$\Big(\tU(t,Y^\nu_t)+\int_0^t\tV(s,Y^\nu_s)ds\Big)$ is a  submartingale, and there exists a unique optimal process $Y^*:=Y^{\nu^*}$  with $\nu^* \in \sigoR$ such that  $\Big(\tU(t,Y^*_t)+\int_0^t\tV(s,Y^*_s)ds\Big)$ is a local martingale.\\
\rmii To summarize, $\tU(t,y)$ is the value function of optimization problem with myopic optimal strategy, that is for any maturity $T\geq t$
\begin{equation}\label{myopic dual value function}
\tU(t,y)=\esssup_{Y^{\nu}\in \GY_t(y)}\E\Big(\tU(T,Y^\nu_T(y))+\int_t^T \tV(s,Y^\nu_s)(y)ds/\F_t \Big), \>a.s. 
\end{equation} 
\rmiii {\sc Optimal Processes characterization} Under regularity assumption, first order conditions imply some links between optimal processes, including their initial conditions,
\begin{eqnarray}\label{first order condition}
 Y^*_t(y)=&U_{x}(t,X^*_t(x))=  V_{c}(t,C^*_t(c)), \quad &y=u_x(x)=v_c(c)
\end{eqnarray} 
The optimal consumption process $C^*_t(c)$ is related to the optimal portfolio $X^*_t(x)$ by the progressive monotonic process $\zeta_t^*(x)$ defined by 
\begin{eqnarray}\label{optimalconsumption}
c=\zeta(x)=-\tv_y(u_x(x)), \quad \zeta_t^*(x)=-\tV_y(U_x(x)), \quad C^*_t(c)=\zeta_t^*(X^*_t(x))
%C^*_t(\zeta(x))=\zeta_t^*(X^*_t(x)), \quad  \zeta_t^*(x)=-\tV_y(U_x(x)), \quad c=\zeta(x)=-\tv_y(u_x(x))
\end{eqnarray} 
\rmiv By Equation \eqref{first order condition}, strong consistency of $(U,V, X^*)$ implies the monotony of $y\mapsto 
Y^*_t(y)$. The system $(\tilde{U},\tilde{V},Y^*)$ is strongly $\GY$-consistent.
\end{Proposition}
\noindent The main consequence of the strong consistency is to provide a closed form for consumption consistent utility system.
\begin{Theorem}\label{explicite construction}
Let $\bar \zeta_t(x)$ be a positive progressive process, increasing in $x$ and let $\bar X_t(x)$ be a strictly monotonic solution with inverse $\bar \X_t(z)$ of the SDE,
$$ 
d \bar X_t(x)= \bar X_t(x) [r_tdt +  \kappa^*_t(\bar X_t(x))(dW_t+\eta^\sigR_tdt)] -\bar \zeta_t(\bar X_t(x))\,dt,\>\>\kappa^*_t(\bar X_t(x))\in \sigR_t.
$$
Let $\bar Y_t(y)$ be a strictly monotonic solution with inverse $(\bar \Y_t(z))$ of the SDE
$$
d \bar Y_t(y)=  \bar Y_t(y)\big[ -r_t dt + (\nu^*_t(\bar Y_t(y))- \eta^{\sigR}_t). dW_t\big],\quad \nu^*_t(\bar Y_t(y))\in \sigoR_t.
$$
Given  a deterministic utility system $(u,v)$ such that $\bar \zeta_0(x)=\zeta(x)=-\tv_y(u_x(x))$, there exists a $\GX^c$-consistent progressive  utility  system $(U,V)$ such that $(\bar X_t(x),\bar Y_t(y))$ are the associated optimal processes,
 defined by:
\begin{equation}\label{charac}
U_x(t,x)={\bar Y}_t(u_x({ \bar\X}_t(x)), \>  V_c(t,c)=U_x(t,{\bar \zeta}^{-1}_t(c)) \>\text{\rm with}
\>\bar \zeta^{-1}(c)=u_x^{-1}(v_c(c))
\end{equation}
\end{Theorem}
\noindent Observe that the consumption optimization contributes only through the conjugate $\tilde V$ of the progressive  utility $\bf V$.
We refer to \cite{MrNek04} for detailed proofs.
\subsection{$\GX^c$-consistent utilities with linear optimal processes}
\noindent The simplest example of monotonic process is given by  linear processes with positive (negative) stochastic coefficient. It is easy to characterize consumption consistent utility sytems $(U,V)$ associated with linear optimal processes
\begin{eqnarray*}
 X^*_t(x)=xX^*_t \>\text{ with  }  \>X^*_t:=X^*_t(1),\qquad Y^{*}_t(y)=yY^*_t \> \text{ with  } \>Y^*_t:=Y^*_t(1)
\end{eqnarray*}
\begin{Proposition} \label{consforwardlinear} 
\rmi A strongly $\GX^c$-consistent progressive  utility $(U,V)$ generates linear optimal
wealth and state price processes if and only if it is  of the form\\
\begin{equation*}
U(t,x)=Y^*_t X^*_t\,u(\frac{x}{X^*_t}), ~ V(t,c)= \zeta_t^* U(t,\frac{c}{\zeta_t^*})\quad \text{ with } {\zeta}_t(x)=x
\>{ \zeta_t}.
\end{equation*}
%with a consumption process \
The  optimal processes are  then given by 
$$ X^*_t(x)=x X^*_t, \>Y^*_t(y)=y Y^*_t, \>\text{and}\quad C^*_t(x)=X^*_t(x)\zeta_t^*.$$
\rmii {\sc Power utilities} A consumption consistent progressive power utility  $(U^\theta,V^\theta)$ (with risk aversion
coefficient
$\theta$) generates necessarily linear optimal processes and is, consequently, of  the form
$\b(U^{(\theta)}(t,x)=\frac{Y^*_tX^*_t}{1-\theta}\b(\frac{x}{X^*_t}\b)^{1-\theta}\>,
V^{(\theta)}(t,x)=(\hat \psi_t)^\theta \frac{Y^*_tX^*_t}{1-\theta}\b(\frac{x}{X^*_t}\b)^{1-\theta}\b)$.
\end{Proposition}
\begin{proof} If $X^*_t(x)=xX^*_t $ and $Y^{*}_t(y)=yY^*_t$, their inverse flows are also linear and
$\X(t,x)=\frac{x}{X^*_t},\Y(t,y)=\frac{y}{Y^*_t}$. \\
 \rmi \rma The linearity with respect to its initial condition of the solution of one dimensional SDE with drift $b_t(x)$ and
diffusion coefficient $\sigma_t(x)$ can be satisfied only when the coefficients $b_t(x)$ and $\sigma_t(x)$ are affine in $x$,
that is $b_t(x)=x b_t\,$ and 
$\sigma^i_t(x)=x \sigma_t^{i} \,$,   $b$ and $\sigma^{i}$ being one dimensional progressive processes. 
Since the only coefficient with some non linearity in the dynamics of $Y^{*}_t(y)$ is $y \nu^*_t(y)$, the previous condition implies that $\nu^*_t(y)$ does not depend on $y$. By the same argument, we see that $x\kappa^*_t(x)$ is linear and $\kappa^*_t(x)$ also does not depend on $x$. For the consumption process, $\zeta^*_t(x)$ the linear condition becomes  $\zeta^*_t(x)= x \zeta_t^*.$\\
\rmb We are concerned by strongly consistent progressive utilities, since optimal processes are monotonic by definition.
Then, since
$Y^*_t(u_x(x))=U_{x}(t,X^*_t(x))$,  we see that the marginal utility $U_{x}$ is given by $U_{x}(t,x)=u_x(x/X^*_t)Y^*_t$. By taking the primitive with the condition $U(t,0)=0,$ $U$ is given by $U(t,x)=u(x/X^*_t)X^*_tY^*_t$.\\
\rmc We know that $C^*_t(x)=X^*_t(x){ \zeta_t^*}$ and
${\bar \zeta_t}(t,X^*_t(x))=-\tV_y(t,U_x(t,X^*_t(x)))$ (from optimality conditions).
Thus $-\tV_y(t,U_x(t,X^*_t(x)))=X^*_t(x){ \zeta^*}_t,\text{ a.s. } \forall t,x.$
From monotonicity of $X^*$ and $U_x$, we then conclude that 
$V_c(t,c)=U_x(t,\frac{c}{\zeta^*_t}) \ \text{ a.s. } \forall t,c.$
Integrating yields the desired formula.\\
\rmii Power-type utilities generate linear optimal processes. So, we only have to consider  initial power utilities $u(x)=k x^{1-\theta}/(1-\theta), v(c)=k_v c^{1-\theta}/(1-\theta)$ with the same risk-aversion coefficient $\theta$ to characterize the system.
\end{proof}
\begin{Remark}
In order to separate the messages and as the  risk aversion  does not vary in this result, we have deliberately
omitted the indexing of the  optimal  process by $ \theta $, especially in the explicit case of  power utilities.
Although, optimal process may reflect a part of this risk aversion, therefore in the last section of this work, we take
care to make them also dependent on this parameter
%, especially because $ \theta $ will not be constant is this last section.
\end{Remark}

\subsection{ Value function of backward classical utility maximization problem}\label{transition:classic} 
As for example in  the Ramsey rule, utility maximization problems in the economic literature use classical utility functions.  
This subsection points out the similarities and the differences between consistent progressive 
utilities and backward classical value functions, and their corresponding portfolio/consumption optimization problems. \\[-8mm]
\paragraph{Classical portfolio/consumption optimization problem and its conjugate problem}
The classic problem of optimizing consumption and terminal wealth  is determined by a  fixed  time-horizon $T_H$ and two deterministic utility functions $u(.)$ and $v(t,.)$  defined up to this horizon.
Using the same notations as  previously, the classical optimization problem  is formulated as the following maximization problem,
\begin{eqnarray}\label{pbopticlassic}
\sup_{(\kappa,c) \in  \GX^c  }  \mathbb{E} \Big( u(X^{\kappa,c}_{T_H})+\int_{0}^{T_H} 
v(t,{c^\mathbb{}_t}) dt \Big).
\end{eqnarray}
For any $[0, T_H]$-valued  $\mathbb{F}$-stopping $\tau$  and  for any positive random variable 
$\mathcal{F}_{\tau}$-mesurable $\xi_{\tau}$,  $ \GX^c(\tau,\xi_{\tau})$ denotes the set of admissible strategies  starting at time $\tau$ with an initial positive wealth  $\xi_{\tau}$, stopped when the wealth process reaches 0.
The corresponding value system (that is a family of random variables indexed by $(\tau, \xi_{\tau})$) is defined as, 
\begin{eqnarray}\label{pbopticlassic}
\mathcal{U}(\tau,\xi_{\tau})= \esssup_{(\kappa,c) \in  \GX^c(\tau,\xi_{\tau})  } \mathbb{E}  \Big( u(X^{\kappa,c} _{T_H}(\tau,\xi_{\tau}))+\int_\tau^{T_H}
v(s,{c^\mathbb{}_s}) ds | {\cal
F}_\tau \Big), \>a.s.
\end{eqnarray}
with terminal condition $\mathcal{U}(T_H,x)=u(x)$. \\
 We assume the existence of  a  progressive utility still denoted $\mathcal{U}(t,x)$ that  aggregates  these system (that  is more or less implicit in the literature). When the dynamic programming principle holds true, the utility system $(\mathcal{U}(t,x), v(t,.)$ is $\GX^c $-consistent.
 Nevertheless, in the backward point of view, it is not easy to show the existence of optimal monotonic processes, or equivalently the strong consistency.
Besides, the optimal strategy in the backward formulation is not myopic and depends on the time-horizon $T_H$. In the economic literature, $T_H$ is often
taken equal to $+\infty$ and the utility function is separable in time with exponential decay at a rate $\beta$ interpreted as the pure time preference parameter: $v(t,c)= e^{-\beta t } v(c)$. 
It is implicitly assumed that such utility function are equal to zero when $t$ tends to infinity. 
%===========================================================

\section{Ramsey rule and Yield Curve Dynamics}
\noindent As our aim is to study long term affine yields curves,  we will focus in the following on  affine optimal processes. But let us first recall  some results on the Ramsey rule with progressive utility. 
\subsection{Ramsey rule }
Financial market cannot give a satisfactory answer for the modeling of long term yield curves,   since  for longer maturities, the bond market is highly illiquid and standard financial interest rates models cannot be easily extended. \\[-9mm]
\paragraph{Economic point of view of Ramsey rule } Nevertheless, an abundant literature
on the economic aspects of long-term policy-making has been developed.  
The Ramsey rule is the reference equation in the macroeconomics literature for the computation of long term discount factor. 
The Ramsey rule comes back to the seminal paper of Ramsey \cite{Ramsey}  in 1928  where  economic interest rates   are  linked with the marginal utility of the aggregate consumption at the economic equilibrium.
More precisely, the economy is represented by the strategy of a risk-averse representative agent, whose utility function on 
consumption rate at date $t$ is the  deterministic function $v(t,c)$. 
 Using an equilibrium point of view with infinite horizon, the Ramsey rule   connects at time $0$ the equilibrium rate for maturity $T$ with the
marginal utility $v_c(t,c)$ of the random exogenous optimal consumption rate  $(C^e_t)$ by
\begin{equation}\label{RamseyRule}
R^e_0(T)=-\frac{1}{T}\ln\frac{\E[v_c(T,C^e_T)]}{v_c(c)}.
\end{equation}
Remark that the Ramsey rule in the economic literature relies on a backward formulation with infinite horizon, an  usual
setting is to assume separable in time  utility function with exponential decay at rate $\beta>0$ and constant risk aversion
$\theta, (0<\theta < 1)$, that is $v(t,c)=Ke^{-\beta t}\frac{c^{1-\theta}}{1-\theta}$.  $\beta$ is the pure time preference
parameter, i.e. $\beta$ quantifies the agent preference of immediate goods versus future ones.  $C^e$ is exogenous and is
often modeled as a geometric Brownian motion.  

In the financial point of view we adopt here  the agent may  invest in a financial market in addition to the money market.
We consider an arbitrage approach with exogenously given interest rate,  instead of an equilibrium approach that determines them endogenously. 
It seems also essential for such maturity to adopt a sequential decision scheme that
allows to revise the first decisions in the light of new knowledge and direct experiences:
the utility criterion must be adaptative and adjusted to the information flow. That is why we consider consistent progressive utility. 
The financial market is 
 an incomplete It\^o  financial  market: notations are  the one described in Section \ref{descmarket},  with a $n$ standard Brownian
motion $W$, a (exogenous) financial short term interest rate  $(r_t)$ and a $n$-dimensional risk premium $(\eta^\sR_t)$. In
the following,
we adopt a financial point of view and consider either the progressive or the backward formulation for the optimization problem. \\[-9mm]
\paragraph{Marginal utility of consumption and  state price density process\\}
\noindent \rmi{\sc The forward dynamic utility problem}\\
Proposition \ref{Duale} gives a  pathwise relation  between the marginal utility of the optimal consumption and the optimal state price density process, where the parameterization is done through the initial wealth $x$, or equivalently $c$ or $y$ since
$c =  -\tv_y(u_x(x))=  -\tv_y(y) $,
\begin{equation}\label{Ramseypathdynamic}
V_c(t, C^{*}_t(c)) = Y^{*}_t(y),   \quad    t \geq 0     \quad   \mbox{ with  } \quad v_c(c) =  y.
 \end{equation}
 The forward point of view emphasizes the key rule played by the monotony of $Y$ with respect to the initial condition $y$,
under 
regularity conditions of the progressive utilities (cf \cite{MrNek01}). 
Then as function of $y$, $c$ is decreasing, and $C^{*}_t(c)$ is an increasing function of $c$.
This question of monotony is frequently avoided, maybe because with power utility functions (the example often used in the
literature)
$Y^{*}_t(y)$ is linear in $y$  as $\nu^*$ does not dependent on $y$. We shall come back  to that issue in Section \ref{sect:mixture}. \\
\rmii{\sc The backward classical optimization problem}\\
  In the  classical optimization  problem,  
both utility functions for terminal wealth and   consumption rate are deterministic, and a given horizon $T_H$ is fixed.  In this backward point of view,  
optimal processes are depending on the time horizon $T_H$ : in particular  the optimal  consumption rate $C^{*,H}(y)$ depends on the time horizon $T_H$ through the optimal state price density process $Y^{*,H},$ leading to the same pathwise relation (\ref{Ramseypathdynamic}) as in the forward case, 
\begin{equation}\label{Ramseypath}
 \frac{v_c(t, C^{*,H}_t(c))}{v_c(c)} =  \frac{Y^{*,H}_t(y)}{y},   \quad    0\leq t\leq {T_H}       \quad   \mbox{ with  } \quad {v_c(c)} =  y.
\end{equation}
So, in general the notation of the forward case are used, but
 with the additional symbol $H$ ($Y^{*,H}, c^{*,H}, X^{*,H}$) to address the dependency on $T_H$ in  the classical backward
problem. \\[1mm]
{\bf Conclusion:} Thanks to the pathwise relation (\ref{Ramseypathdynamic}), the Ramsey rule yields to a description of the equilibrium interest rate as a	 function of the optimal  state price density process $Y^{*,e}$, $R^e_0(T)(y)=-\frac{1}{T}\ln \mathbb{E}[Y^{*,e}_T(y)/y]$, that allows to give a financial interpretation in terms of zero coupon bonds. More dynamically in time,  $\forall t < T,$%with $Y^{*}_{t,T}(y) :=Y^{*}_T(y)/Y^{*}_t(y), $
%$Y^{*}_{t,T}(y) :=\frac{Y^{*}_T(y)}{Y^{*}_t(y)}, $
\begin{equation}\label{Ramseyfin}
R^e_t(T)(y):=-\frac{1}{T-t} \ln \mathbb{E}\left[\frac{V_c(T,C^{*,e}_T(c))}{V_c(t,C^{*,e}_t(c))}\big|
\mathcal{F}_t\right]=-\frac{1}{T-t}\ln
\mathbb{E}\left[\frac{Y^{*,e}_T(y)}{Y^{*,e}_t(y)}\b| \mathcal{F}_t\right].
\end{equation}
%
%
%
%============================================================================
\subsection{Financial yield curve dynamics}
 Based on the foregoing,
it is now proposed to make  the connection between the economic and the financial point of view through the state price
densities processes and the pricing.\\
 Let  $\b(B^m(t,T),\>t\leq T\b)$, ($m$ for market), be the market price  at time $t$ of a  zero-coupon bond paying one unit of cash at maturity $T$. Then, the market yield curve is defined as usual by the actuarial relation, $B^m(t,T)=\exp(-R^{m}_t(T)(T-t)).$
Thus our aim is to give a  financial interpretation of $\mathbb{E}\Big[\frac{Y^{*}_T(y)}{Y^{*}_t(y)}\b| \mathcal{F}_t\Big]$ for $ t\leq T$ in terms of price of zero-coupon bonds.\\
 Remark that $Y^{*}$ is  solution of an optimization problem whose criteria depend on the utility functions, yet the utilities do not intervene in the dynamics of $Y^*$. 
In term of pricing, the terminal wealth at maturity $T$ represents the payoff at maturity of the financial product, whereas the consumption may be interpreted
as the dividend distributed by the financial product before $T$. 
%The zero-coupon bond paies one unit of cash at maturity and does not distribute any dividends. 
%
\paragraph{Replicable bond}
For admissible portfolio without consumption $X^\kappa_t$,  it is straightforward that for any state price process $Y^\nu_tX^\kappa_t$ is a local martingale, and so under additional integrability assumption, $ X^\kappa_t=\mathbb{E}\big[X^\kappa_T \frac{Y_T^\nu}{Y^\nu_t}\b|\mathcal{F}_t\big]$. So the price of $X^\kappa_T$ does not depend on $\nu$. This property holds true for any derivative whose the terminal value is replicable by an admissible portfolio without consumption, for example a replicable bond,
$$B^m(t,T)=\mathbb{E}\big[\frac{Y_T^0}{Y^0_t}\b|\mathcal{F}_t\big]=\mathbb{E}^{\Q}\b[e^{-\int_t^Tr_s
ds}\b|\mathcal{F}_t\big]=\mathbb{E}\Big[\frac{Y^{}_T(y)}{Y^{}_t(y)}\b| \mathcal{F}_t\Big]$$
Besides, $B^m(t,T)=\mathbb{E}\Big[\frac{Y^{}_T(y)}{Y^{}_t(y)}\b| \mathcal{F}_t\Big]$ for any state price density process $Y$
with goods integrability property.
\\[-8mm]
%Moreover, %if $\bar{X}$ denotes an admissible wealth (without consumption)  which replicates  $1$ at the maturity $T$, then
%the bond $B^m(t,T)$  may be considered as the  price at time  $t$ of any admissible portfolio $\bar{X}$  with terminal wealth
%$\bar{X}_{T}=1$, i.e.,  $B^m(t,T)=\mathbb{E}^{\Q}\Big[e^{-\int_t^Tr_sds}\bar{X}_T\b| \mathcal{F}_t\Big]$.
\paragraph{Non hedgeable  bond}
For non hedgeable zero-coupon bond, the pricing by indifference is a way (among others) to evaluate the risk coming from the unhedgeable part. \\
The {\em utility indifference price} is the cash amount for which the investor is indifferent between selling (or buying) a given quantity of the claim or not. 
This pricing rule is non linear and provides a bid-ask spread. If the investor is aware of its sensitivity to the unhedgeable risk, they can try to transact
for  a little amount. In this case, the ``fair price'' is the marginal utility indifference  price  (also called Davis price
\cite{Davis}), it
 corresponds to the  zero marginal rate of substitution. 
 We denote by $B^{u}(t,T)$ ($u$ for utility) the marginal utility price at time $t$ of  a zero-coupon  bond paying one cash unit at maturity $T$, that is $\quad B^{u}(t,T)=B^{u}_t(T, y)=  \mathbb{E}\B[ \frac{Y^{*}_{T}(y)}{Y^*_t(y)} \b| \mathcal{F}_t  \B]$. Based on the link between optimal state price density and optimal consumption, we see that
\begin{equation}\label{linkYyieldcurve}
B^{u}_t(T,y):=B^{u}(t,T)(y)=   \mathbb{E}\B[ \frac{Y^{*}_{T}(y)}{Y^*_t(y)} \b| \mathcal{F}_t  \B]=\mathbb{E}
\B[\frac{V_c(T,C^{*}_T(c))}{V_c(t,C^*_t(c))} \b| \mathcal{F}_t\B], \quad v_c(c)=y.
\end{equation}
Remark that $B^{u}(t,T)(y)$ is also equal to $\mathbb{E}
\B[\frac{U_x(T,X^{*}_T(x))}{U_x(t,X^*_t(x))} \b| \mathcal{F}_t\B]$. Nevertheless, besides the economic interpretation, the formulation through the 
optimal consumption is more relevant than the formulation through the 
optimal wealth : indeed  the utility from consumption $V$ is given, while the utility from wealth  $U$ is more constrained. \\
According to the Ramsey rule (\ref{Ramseyfin}), equilibrium interest rates and marginal utility interest rates are the same. Nevertheless, for  marginal utility price, this last curve is robust only for small trades.\\
The martingale property of $Y^*_t(y) B^{u}_t(T,y)$  yields to the following dynamics  for the zero coupon bond maturing at time $T$ with  volatility vector $\Gamma_t(T,y)$ %is the volatility vector of the zero-coupon bond $B^{u}(t,T)(y)$ maturing at  $T$,
\begin{equation}
\frac{dB^{u}_t(T,y)}{B^{u}_t(T,y)}= r_t dt+\Gamma_t(T,y).(dW_t+(\eta^\sigR_t-\nu^*_t(y))dt).
\end{equation}
Using the classical notation for exponential martingale, ${\cal E}_t(\theta)=\exp\b(\int_0^t \theta_s.dW_s-\demi \int_0^t\|\theta_s\|^2.ds\b)$,  the martingale $Y^*_t(y) B^{u}_t(T,y)$  can written as an exponential martingale with volatility $\b(\nu_.^*(y)-\eta^\sigR_.+\Gamma_.(T,y)\b)$. 
Using that $B^{u}_T(T,y)=1$, we have two characterisations of $Y^*_T(y),$
\begin{eqnarray*}
Y^*_T(y)= B^{u}_0(T, y){\cal E}_T\b(\nu_.^*(y)-\eta^\sigR_.+\Gamma_.( T,y)\b)= y\> e^{-\int_0^T r_s ds}
\cal{E}_T\b(\nu_.^*(y)-\eta^\sigR_.\b).
\end{eqnarray*}
Taking the logarithm gives
\begin{equation}\label{eq:intrsds}
\int_0^Tr_s ds=TR^u_0(T)-\int_0^T\Gamma_t(T,y).(dW_t+(\eta_t-\nu^*_t(y))dt)+\demi \|\Gamma_t(T,y)\|^2 dt .
\end{equation}
\subsection{Yield curve for infinite maturity and progressive utilities} 
The computation of the marginal utility  price of zero coupon bond is then straightforward using (\ref{linkYyieldcurve})
%then  for all $0 \leq t  < T \leq  T_H$
% \begin{eqnarray*}
% B^{u,H}(t,T)&=& \exp \left( \text{Cst}(T)-\text{Cst}(t) + \int_0^t (\Gamma_s(T)-\Gamma_s(t)   )dW_s \right)\\
% && \quad \exp \left(\frac{1}{2} \int_t^T (||\Gamma_s(T) ||^2 +  2 <\Gamma_s(T), \nu_s^{*,H}-\eta^\sigR_s   > )   ds
%  \right)\\
% &=  &
% B^{u}(0,T) \exp \left( \int_0^t r_s ds + \int_0^t  \Gamma_s(T) dW_s \right)\\  
% && \exp \left( -\frac{1}{2} \int_0^t (||\Gamma_s(T) ||^2 +  2 <\Gamma_s(T), \nu_s^{*,H}-\eta^\sigR_s   > )   ds
%  \right).
%  \end{eqnarray*}
leading to the yield curve dynamics $(R^{u}_t(T,y)= - \frac{1}{T-t} \ln B^{u}_t(T,y))$
\begin{eqnarray*}
R^{u}_t(T,y)& =&\frac{T}{T-t}
R^{u}_0(T,y)  -  \frac{1}{T-t} \int_0^t r_s ds - \int_0^t  \frac{\Gamma_s(T,y)}{T-t} dW_s  \\
 &&+ \int_0^t \frac{||\Gamma_s(T,y) ||^2}{2(T-t)} ds +  \int_0^t  <\frac{\Gamma_s(T,y)}{T-t}, \nu_s^{*}-\eta^\sigR_s   >    ds.
 \end{eqnarray*}
 for finite maturity, and $l^{u}_t(y):= \lim_{T \rightarrow + \infty}R^{u}_t(T,y) $ for infinite maturity.\\
As showed in Dybvig \cite{Dybvig} and  in El Karoui and alii. \cite{ElKarouiFrachot} the long maturity rate $l^{u}_t(y)$  behaves differently
according to the long term behavior of the volatility when $T\to \infty$,
\bit
\item[$-$] If $\lim_{T\to\infty} \frac{\Gamma_t(T,y)}{T-t}\neq 0, a.s.$, then 
$\frac{||\Gamma_t(T,y)||^2}{T-t}\to \infty$ a.s and $l_t(y)$ is infinite.
\item[$-$] Otherwise, $l_t=l_0+ \int_0^t \lim_{T \to \infty}  \left( \frac{||\Gamma_s(T,y)||^2}{2(T-s)}   \right) ds$, and $l_t$ is a non decreasing process, constant in $t$ and $\omega$ if $\lim_{T \to \infty}  \frac{||\Gamma_t(T,y)||^2}{T-t}=0$.
\eit
{\em In this last case, which is the situation considered by the economists, all past, present or future yield curves have the same asymptote}.
\section{ Progressive utilities and  yield curves in affine factor model}
Recently, affine factor models have been intensively developed with some success to capture under the physical probability
measure both financial and macroeconomics effects, from the seminal paper of Ang and Piazzesi (2003). As explained in
Bolder\&Liu (2007) \cite{Jamieson},
{\em Affine term-structure models have a number of theoretical and practical advantages. One of the principal advantages is
the explicit description of market participants aggregate attitude towards risk. This concept, captured by the market price
of risk in particular, provides a clean and intuitive way to understand deviations from the expectations hypothesis and
simultaneously ensure the absence of arbitrage.}

\subsection{Definition of affine market}
The affine factor model makes it possible to compute tractable pricing formulas, it extends the  log-normal model (studied in
\cite{MrNek04})  to a more stochastic model. 
Affine model,  which generalizes the CIR one, was  first introduced by D. Duffie and R. Kan (1996)
\cite{Duffie}, where the authors   assume that the yields are affine function of stochastic factors, which implies  an affine
structure of the factors.  
Among many others,  M.
Piazzesi reports   in  \cite{Monika} some recent successes in the study of affine term structure models. 
Several constraints must be fulfilled to define an affine model in mutidimentional framework, but we will not discuss the
details here and refer to the works of Teichmann and coauthors \cite{Teichmann2},  \cite{Teichmann1}.

\paragraph{Properties of affine processes and their exponential}
We adopt the framework of the example in 
Piazzesi (\cite{Monika}, p 704). 
The factor is a N-dimensionnal vector process denoted by   $\xi$ and  is  assumed to be  an affine diffusion process, that is
the drift
coefficient and the variance-covariance matrix are affine function of $\xi$ :
\begin{equation}\label{dynamicxi}
d\xi_t=\delta_t(\xi_t)dt+\sigma_t(\xi_t)dW_t
\end{equation} 
The affine constraint is expressed as:
\bit
\item[-] $\delta_t(\xi_t)=\varrho_t^\delta \xi_t+\delta_t^0$ , where $\varrho_t^\delta \in \R^{N\times N}$ and $\delta_t^0
\in
\R^{N}$ are  deterministic.
\item[-] $\sigma_t(\xi_t)=\Theta_t s_t(\xi_t)$, where $\Theta_t \in\R^{N\times N}$ is deterministic, and the matrix
  $s_t(\xi_t)$ is a diagonal $N\times N$ matrix, with eigenvalues $s_{ii,t}(\xi_t).$
 The affine property concerns the variance covariance matrix $\Theta_t s_t(\xi_t)\tilde s_t(\xi_t)\tilde \Theta_t$ or
equivalently (since  $s_t(\xi_t)$ is diagonal) the positive eigenvalues 
  $\lambda_{i,t}=s^2_{ii,t}(\xi_t)$ of $s_t(\xi_t)\tilde s_t(\xi_t)$:   $\lambda_{i,t}=\tilde{\rho}^{\lambda}_{i,t}\xi_t+
\lambda^0_{i,t}$ 
{ that must be positive \footnote{This constraint is restrictive,  but we do not go into the details of this
hypothesis as this is not the goal of this work, for
more details see the  literature cited above concerning affine models. }} with deterministic
$(\lambda^0_{i,t},\tilde{\rho}^{\lambda}_{i,t})\in \R\times\R^N$, where  $\, \tilde{.} \,$  denotes the transposition of a
vector or a matrix.
\eit

\paragraph{Characterization of market with affine optimal processes}
 To be coherent with the previous market model (Section \ref{descmarket}), we have to define the set of admissible strategies
$\sigR_t$ ,  at date
$t$, and its  orthogonal $\sigoR_t$. Let us first  to point out that the $(N\times1)$ volatility vector of any process
$\tilde{a}_t\xi_t$ ($a$ is
deterministic) is given by $s_t(\xi_t)\tilde{\Theta}_t.a_t$. Thus if $\sigR$ (linear space) is the
set
of admissible strategies, then at time $t$ it depends on $\xi_t$ and is  necessarily  given by, 
$$\sigR_t(\xi_t)=\{\tilde a_t  \Theta_ts _t(\xi_t), ~a_t \>\text{ progressive vector}\>\mbox{ in some linear and 
deterministic space } E_t\subset\R^N \}$$
The deterministic space $E_t$ and its orthogonal are assumed to be stable by  $\tilde \Theta_t$,  or equivalently  the matrix
  $\tilde
\Theta_t$ is commutative with the orthogonal projection on $E_t$.   A block matrix $\Theta_t$ (up to an orthogonal
transformation) satisfies this property. Furthermore, $E_t$ and $E^\perp_t$ are stable by  $\varrho^{\delta}_t$.
 The set of admissibles strategies being well defined, we denote by $a_t^\sigR$ and  by $a_t^\perp$ the elements of the
linear
spaces $E_t $ and $E^\perp_t$.\\
We consider two types of assumptions:\\
\rmi The spot rate $(r_t)$ and the consumption rate $(\zeta_t)$ are  affine positive processes \\
\centerline{$r_t= a_t^r \xi_t + b_t^r\quad $ and $\quad \zeta_t= a_t^\zeta \xi_t + b_t^\zeta$.}\\
\rmii The volatilities of the optimal processes  $X^*$ and $Y^*$ have  affine structure,
\begin{equation}\label{eqvolaffine}
\kappa^*_t=\tilde{a}^{X,\sigR}_t\Theta_t s_t(\xi), \quad     \eta^{\sigR}_t =\tilde{a}^{Y,\sigR}_t\Theta_t
s_t(\xi),
\quad \nu^*_t=\tilde{a}^{Y,\perp}_t\Theta_t s_t(\xi), (a^{.,\sigR} \in E_t \mbox{ and }  a^{.,\perp}\in E^\perp_t).
\end{equation}
%Both local martingales $ (e^{\int_0^t r_s ds} Y^* _t)$ and $ (e^{-\int_0^t \zeta_s ds} Y^*_tX^*_t)$ are  exponential functions
%of the processes
%$\tilde{a}^{Y}_t(\xi_t-\xi_0)+\int_0^t b^Y_s(\tilde{a}^{Y}_s,\xi_s)ds$ and
%$(\tilde{a}^{Y}_t+\tilde{a}^{X}_t)(\xi_t-\xi_0)+\int_0^t b^{XY}_s(\tilde{a}^{Y}_s+\tilde{a}^{X}_s,\xi_s)ds$, where
%$b^Y_s(\tilde{a},\xi)$ and $b^{XY}_s(\tilde{a},\xi) $ are quadratic affine functions in $\tilde{a}$ and affine in $\xi$ (see
%Lemma \ref{LemmaMarAff} below).

%==============================
\paragraph{Forward utility and marginal utility yields curve} 
\noindent In the forward utility framework with linear portfolios, the marginal utility price of zero-coupon bond with
maturity $T$ is given by Equation \eqref{linkYyieldcurve},  where $Y^*_t(y)$ is linear in $y$. The price of zero-coupon does
not depend on $y$. Taking into account  the specificities of the affine market, we have that
\begin{eqnarray}\label{affineYyieldcurve}
&&B^{u}_t(T)=\mathbb{E}\big[ \frac{Y^{*}_{T}}{Y^*_t} \b| \mathcal{F}_t  \big]\\
&&=\mathbb{E}\Big[\exp\big(-\int_t^T (a_s^r \xi_s + b_s^r) ds+\int_t^T\tilde{a}^{Y}_u\Theta_u s_u(\xi)dW_u-\demi \int_t^T
\|\tilde{a}^{Y}_u\Theta_u s_u(\xi)\|^2 du\big)\big| \mathcal{F}_t  \Big]\nonumber
\end{eqnarray}
Thanks to the Markovian structure of the affine diffusion $\xi$, the price of zero-coupon bond is an exponential affine
fonction of  
$\xi$, with the terminal constraint $B^{u}_T(T)=1$, \\[1mm]
\centerline{$\ln(B^{u}_t(T))=\tilde A_t^T \xi_t +B_t^T, \quad A_T^T=0, \>B_T^T=0$.}\\[1mm] 
 {\sc \bf Ricatti equations}
To justify the Ricatti equations, we fix some notations and develop useful calculation for study exponential affine process.
Moreover, when working with affine function $f(t, \xi)=\tilde a_t \xi +b_t$,  it is sometimes useful to write
$b_t=f_t(0)=f_t^0$ and $a_t=\nabla_\xi f_t(0)=\nabla f^0_t.$
\begin{Lemma}\label{LemmaMarAff}
 Let $a_t\in\R^N$ and $b_t\in \R$ be two deterministic functions. \\
 \rmi The affine process $\tilde{a}_t\xi_t+b_t$ is a semimartingale with
decomposition,
\begin{equation*}
d(\tilde{a}_t\xi_t+b_t)=\tilde{a}_t\Theta_ts_t(\xi)dW_t -\frac{1}{2}||\tilde{a_t}\Theta_t s_t(\xi)||^2 dt +f_t(a_t,\xi_t) dt
\end{equation*}
\rma The quadratic variation $||\tilde{a}_t\Theta_ts_t(\xi)||^2=q_t(a,\xi_t)$ is quadratic in $a$ and affine in $\xi$. More
precisely, if $\Theta^i_t$ is the $i$-th column of the matrix $\Theta$,  we have 
\begin{eqnarray}\label{driftaffine}
q_t(a,\xi)=\nabla q_t^{0}(a) \xi+q^0_t(a), \>  \text{with}\> \nabla
q_t^{0}(a)=\Sigma_{i=1}^{i=N}\big(\tilde{\Theta}^i_ta_t\big)^2\rho^{\lambda}_{i,t}, \> \text{and} \>
q_t^0(a)=\Sigma_{i=1}^{i=N}\big(\tilde{\Theta}^i_ta_t\big)^2{\lambda}^0_{i,t}\nonumber
\end{eqnarray}
\rmb The drift term $f_t(a, \xi_t)$ is an affine quadratic form in $a$ and affine in $ \xi$ given by
\begin{eqnarray}\label{driftaffine}
f_t(a,\xi_t) =(\partial_t \tilde{a}_t+\tilde{a}_t \varrho^{\delta}_t  +\frac{1}{2} \,\nabla q_t^{0}(a) \xi_t +\partial_t
b_t+\tilde{a}_t\delta^0_t+\frac{1}{2} \,q_t^{0}(a) \end{eqnarray}
\rmii A process $X_t=\tilde{a}_t\xi_t+b_t+\int_0^t \delta_s^X(\xi_s) ds$ ($\delta_t^X(\xi)=\nabla \delta_t^{0,X} \xi
+\delta^{0, X}_t$) is the log of an exponential martingale if and only if the coefficients satisfy the Ricatti equation
\begin{equation}\label{Ricatti equation}
\partial_t \tilde{a}_t+\tilde{a}_t \varrho^{\delta}_t  + \frac{1}{2} \,\nabla q_t^{0}(a) +\nabla \delta_t^{0,X}=0, \quad \partial_t
b_t+\tilde{a}_t\delta^0_t+\frac{1}{2} \,q_t^0(a)+\delta^{0,X}_t=0
\end{equation}
\end{Lemma}
%==============================
\begin{proof}
\rmi From It\^o's formula,
$d(\tilde{a}_t\xi_t+b_t)=\partial_t (\tilde{a}_t)\xi_tdt +\tilde{a}_td\xi_t+\partial_tb_tdt$,\\
which implies from the dynamics (\ref{dynamicxi}) of $\xi$,
\begin{eqnarray*}
&&d(\tilde{a}_t\xi_t+b_t)=\big(\partial_t (\tilde{a}_t)\xi_t
+\tilde{a}_t\delta_t(\xi_t)+\partial_tb_t\big)dt+\tilde{a}_t\Theta_t
s_t(\xi)dW_t\\
&&=\Big(\partial_t (\tilde{a}_t)\xi_t
+\tilde{a}_t\delta_t(\xi_t)+\partial_tb_t+\frac{1}{2}||\tilde{a}_t\Theta_ts_t(\xi)||^2\Big)dt
+\Big(\tilde{a}_t\Theta_t
s_t(\xi)dW_t-\frac{1}{2}||\tilde{a}_t\Theta_ts_t(\xi)||^2dt\Big).
\end{eqnarray*}
\rma The sequel is based on the decomposition of the quadratic variation of $\tilde{a}_t\xi_t+b_t$ as affine form in $\xi_t$
and quadratic in $a$.\\
\rmb Then, the affine decomposition of the drift term may be rewritten in the same way.\\
\rmii $\exp(X_t)$ will be an exponential martingale, if and only if $dX_t=\tilde{a}_t\Theta_ts_t(\xi)dW_t
-\frac{1}{2}||\tilde{a_t}\Theta_t s_t(\xi)||^2 dt, $ 
or equivalently if and only if  $f_t(a,\xi) + \delta^X(t,\xi) \equiv 0$. Given that $f_t(a,\xi) + \delta^X(t,\xi) $ is affine
in $\xi$, the condition is satisfied if both, the coefficient of $\xi$ and the constant term are null functions.
\end{proof}
%======================================================
\noindent {\bf Application to  bond pricing} Solving Ricatti Equation \eqref{Ricatti equation} with any given  initial
condition $\tilde a_0 \xi_0+b_0$ implies that the exponential affine process $\exp(\tilde{a}_t\xi_t+b_t+\int_0^t
\delta^X_s(\xi_s) ds)$ is a local martingale. \\
 In a backward formulation,  solving Ricatti Equation \eqref{Ricatti equation} with   terminal constraint $X_T=\tilde{a}_T
\xi_T+b_T+\int_0^T\delta^X_s(\xi_s) ds$ implies under additional integrability assumption, that, if $(a^T_t, b_t^T) $ are
solutions of the Riccatti system with terminal condition $ (a_T, b_T) $.
\begin{equation}\label{condformula}\E\big[\exp(X_T-X_t)|\cal F_t\big]=\exp\big(\tilde{a}^T_t \xi_t+b_t^T \big).
\end{equation}
In the theory of bond pricing, from Equation \eqref{affineYyieldcurve}, we are looking for an  affine process $\tilde{A}^T_t
\xi_t+B^T_t$ such that, since $r_t=a_t^r \xi_t + b_t^r$,
$$\exp\big(\tilde{A}^T_t \xi_t+B^T_t \big)=\E\Big[
\exp
\big(\int_t^T -(a_t^r \xi_t + b_t^r )du+\tilde{a}^{Y}_u\Theta_u s_u(\xi)dW_u-\demi \|\tilde{a}^{Y}_u\Theta_u s_u(\xi)\|^2
du\big)
|\mathcal{F}_t \Big]$$
Given that the martingale part is the one associated with the affine process $\tilde{a}^{Y}_t \xi_t$, it is possible to
define an affine process $X^Y$  with affine drift such that, 
\begin{equation} 
X^Y_t=\tilde{a}^{Y}_t \xi_t-\int_0^t f_s(a^{Y}_s, \xi_s)\,ds=X^Y_0+\int_0^t\tilde{a}^{Y}_u\Theta_u s_u(\xi)dW_u-\demi
\|\tilde{a}^{Y}_u\Theta_u s_u(\xi)\|^2 du
\end{equation} 
where $f_t(a, \xi)$ is defined in Equation \eqref{driftaffine} with $b=0$. Since $\exp(X^Y_t)$ is a local martingale,
$\tilde{a}^{Y}_t$ is solution of the Ricatti equation with $b_t\equiv 0$.
In the new formulation, we are looking for some process $X^T_t=\tilde{A}^T_t \xi_t+B^T_t $ such that the exponential of the
process $Z_t= X^T_t+X^Y_t-\int_0^t\,r_sds$ is a local martingale.
The problem belongs to the family of previous problem applied to the process $Z$.
We summarize these results below.
%======================================
\begin{Theorem} Assume an affine optimization framework, where the optimal state price has an affine volatility
$\tilde{a}^{Y}_t\Theta_t s_t(\xi)$ with $\tilde{a}^{Y}_t$ solution of some Ricatti equation. \\
\rmi Any zero-coupon bond is an exponential function $\exp\big(\tilde{A}^T_t \xi_t+B^T_t \big)$ such that $\tilde
a^Z_t=\tilde{A}^T_t+\tilde{a}^{Y}_t$ is solution of a Ricatti function with terminal condition $\tilde{a}^{Y}_T$, and
$\delta^Z $ function 
$\delta^Z_t(\xi) =-f_t(a^{Y}_t \xi)+a_t^r \xi_t + b_t^r $.\\
\rmii The volatility of a bond with maturity $T$ is $\Gamma_t(T)=\tilde{A}^T_t \Theta_t s_t(\xi)$.
\end{Theorem} 
\begin{proof} The process $Z_t= X^T_t+X^Y_t-\int_0^tr_sds$ is an affine process with affine integral term. The decomposition
of $Z$ is of the type $Z_t= \tilde{a}^Z_t\xi_t+b^Z_t+\int_0^t  \delta^Z_s(\xi_s) ds$ with
$\tilde{a}^Z_t=\tilde{A}^T_t+\tilde{a}^{Y}_t$, $b^Z_t=B^T_t$, $\delta^Z_t(\xi) =-f_t(a^{Y}_t, \xi)+a_t^r \xi_t + b_t^r $. So
$\tilde{a}^Z_t$ satisfies a Ricatti equation with terminal value $\tilde{a}^{Y}_T.$
\end{proof}

\subsection{Affine model and   power utilities}\label{sec:expuissance}
As power utilities is the classical most important example for economics, we now study the marginal utility yield curve in
affine model with progressive and backward power utilities. The backward case differs significantly of the forward case,
since constraints appear on the optimal processes at maturity $T_H$.\\[-8mm]
\paragraph{Backward formulation}
In the classical backward problem with classical power utility function $x^\theta$, ($0<\theta <1$) and horizon $T_H$,  we
have the terminal constraint on the optimal processes: the terminal values of the   optimal
wealth process $X^{*,H}_{T_H}$ and of the state  price  process  $Y^{*,H}_{T_H}$ satisfy   (from optimality) \\
\centerline{ $(Y^{*,H}_{T_H})^{1/\theta} X^{*,H}_{T_H}=\text{Cst}=\exp k$ }\\
We recall here by the index $H$ the dependency on the horizon $T_H$; moreover for simplicity we make $k=0$. 

The question is then how this constraint is propagated at any time in an affine framework, with affine consumption rate $\zeta_t^*$. For notational simplicity, we denote $X^{*,\zeta,H}$ 
(the optimal wealth process capitalized at rate $\zeta_t^*$ by the process $S^{0,\zeta^*}_t=\exp\int_0^t\zeta_u^*du.$ \\
Using that $S^{0,\zeta^*} X^{*,H}Y^{*,H}$ is a martingale with terminal value  $S^{0,\zeta^*}_{T_H} (Y^{*,H}_{T_H})^{1-1/\theta} $, we study the martingale 
$M^\theta_t=\E\big(S^{0,\zeta^*}_{T_H} (Y^{*,H}_{T_H})^{1-1/\theta}|\cal F_t\big )$ in two ways:\\
 \rmi the first one is based on the process $X_t^{*,H}$, since $M^\theta_t=S^{0,\zeta^*}_tX_t^{*,H}Y_t^{*,H}$\\
 \rmii the second one is very similar to the study of zero-coupon bond, by observing that by the Markov property
$ \E\big[(S^{0,\zeta^*}_{T_H}/S^{0,\zeta^*}_t)(Y^{*,H}_{T_H}/Y^{*,H}_t)^{1-1/\theta}|\cal F_t\big ]$ is an exponential affine, whose  coefficients
$(A^\theta,B^\theta)$ are solutions of a Ricatti equation,   that is
$$\E\big((S^{0,\zeta^*}_{T_H}/S^{0,\zeta^*}_t)(Y^{*,H}_{T_H}/Y^{*,H}_t)^{1-1/\theta}|\cal F_t\big )=\exp({\tilde A}_t^\theta\xi_t+B_t^\theta)$$
The backward constraint is then equivalent to the stochastic equality \\[1mm]
\centerline{$S^{0,\zeta^*}_tX_t^{*,H}Y_t^{*,H}=S^{0,\zeta^*}_t(Y^{*,H}_t)^{1-1/\theta}\exp({\tilde A}_t^\theta\xi_t+B_t^\theta)$.}
\begin{Proposition}
\rmi The terminal optimal constraint $X^{*,H}_{T_H}=S^{0,\zeta^*}_{T_H}(Y^{*,H}_{T_H})^{-1/\theta}$ is propagated  through the time into a
closed relation with the state price process,
\begin{equation}
 X_t^{*,H}=(Y^{*,H}_t)^{-1/\theta}\exp({\tilde A}_t^\theta \xi_t+B_t^\theta),
\end{equation}
where $\exp({\tilde A}_t^\theta\xi_t+B_t^\theta)=\E\big[(S^{0,\zeta^*}_{T_H}/S^{0,\zeta^*}_t)(Y^{*,H}_{T_H}/Y^{*,H}_t)^{1-1/\theta}|\cal F_t\big ]$.\\
\rmii In particular, the constraint that the volatility $\kappa^*_t$ has an affine structure belonging to the space $E$,
implies that $({\tilde A}_t^\theta)^\perp=1/\theta(\tilde a^Y_t)^\perp$.\\
\rmiii The links with the zero-coupon bond is given by the relation, 
\begin{equation}
B(t,T)=\exp({\tilde A}^T_t\xi_t+B_t^T)=\E\big(\frac{Y^{*,H}_{T}}{Y^{*,H}_t}|\cal F_t\big ).
\end{equation}
\end{Proposition}

\section{Yield curve dynamics non-linear  on initial conditions  }\label{sect:mixture}

\subsection{From linear optimal processes to more general progressive utilities }
Until then, we have omitted the dependence of optimal processes with respect to risk aversion. In what follows, risk aversion
plays an important role, therefore, an agent that has as risk aversion denoted $ \theta $, his utility process will be
denoted $ U^\theta $, his optimal wealth is denoted $ X ^ {* \theta} $ and finally  his  optimal dual process $Y^{*,
\theta}$. For simplicity we are concerned  in the following only by utility processes which are of power type. 
As we have already mentioned above, utilities of power type  generate optimal processes $ X ^{*,\theta} $ and $ Y^{*,\theta}
$ which are linear with respect to their initial conditions, i.e, $ X^ {*,\theta} (x) = xX^{*,\theta} $ and $
Y^{*,\theta}(y) = yY^{*,\theta} $ (with $ X ^{*,\theta} = X ^ {*,\theta} (1) $ and $ Y ^ {*,\theta} = Y ^ {* ,\theta} (1))
$. Thus the marginal utility  price at time $t$ of a zero coupon with maturity $T$,
given  in \eqref{linkYyieldcurve}  by
\begin{equation*}
B^{\theta}(t,T)(y):=B^{u}(t,T)(y)=   \mathbb{E}\B[ \frac{Y^{*,\theta}_{T}(y)}{Y^{*,\theta}_t(y)} \b| \mathcal{F}_t  \B]=
\mathbb{E}\B[ Y^{*,\theta}_{t,T}\b| \mathcal{F}_t  \B],~\text{ with } Y^{*,\theta}_{t,t}=1
\end{equation*}
  does not depend on  $ y $  nor  on the  consumption of the market  at  time $ t $.
This is not surprising given that power utilities, although they are   useful to compute  explicit optimal strategies,  are somehow restrictive because they  generates only  
linear optimal processes. Besides, the  economic litterature  emphasizes the  dependence of the equilibrium rate $R_0^e(T)$
on the initial consumption.
%  see \cite{Gollier3,Gollier6,Gollier15,HGlongterm}.\\
To study this dependence, we have to give a nontrivial example of stochastic utility that generates a  nonlinear state price
density process and then calculate the price of zero coupon.
This is not obvious, especially since our goal is to give an explicit formula for the optimal dual process. The idea is
to first generate, from optimal process $ X ^ {*,\theta} $ and $ Y ^ {*,\theta} $ associated with progressive power utilities
$U^\theta$, a new processes $ \bar{X} $ and $ \bar{Y} $ which are  both admissible, monotone and
especially nonlinear with respect to their initial conditions. In a second step,   we use the characterization
\eqref{charac} of Theorem \ref{explicite construction} to thereby construct
non-trivial stochastic utilities with $ \bar{X} $ and $ \bar{Y} $ as optimal processes.
The method that we will develop in the following is the starting point of the work of El Karoui and
Mrad \cite{MrNek03}, in which  many  other ideas and extensions can then be found.\\

\noindent
\underline{{\bf step 1:}} To fix the idea, $(X ^ {*,\theta}, Y^ {*,\theta}, \zeta^{*,\theta} )$ denotes the triplet of
optimal primal, dual  and consumption processes associated with stochastic utility $(U^\theta,V^\theta)$ of power type and
relative risk aversion $U_x^\theta/xU_{xx}^\theta=\theta$. 
We consider also two strictly decreasing probability density functions $f$ and $g$,
$\int_{0}^{+\infty}f(\theta')d\theta'=\int_{0}^{+\infty}g(\theta')d\theta'=1$, use the change of variable $
\theta=z\theta'$ and define the strictly increasing functions  $x^\theta(x):=f(\theta/x),~x>0$,
$y^\theta(y):=g(\theta/y),~y>0$ satisfying the following identities 
$$\int_{0}^{+\infty}x^\theta(x)d\theta=x,~\forall x>0$$
$$\int_{0}^{+\infty}y^\theta(y)d\theta=y,~\forall y>0$$

\noindent
 \underline{{\bf step 2:}} We  are now concerned with the following processes $\bar{X}$, $\bar{Y}$ and $\bar{\zeta}$ defined by
\begin{eqnarray}
 &\bar{X}_t(x):=\int_{0}^{+\infty}x^\theta(x)X^{*,\theta}_td\theta,~ x>0 \label{BarXMixture}\\
& \bar{Y}_t(y):=\int_{0}^{+\infty}y^\theta(y)Y^{*,\theta}_td\theta,~ y>0 \label{BarYMixture}
\end{eqnarray}
where we recall that $X^{*,\theta}_t$ (resp. $Y^{*,\theta}_t$) denotes, in this integral, the optimal process starting
from the initial condition equal to $1$.
As seen previously, these two processes are an admissible wealth and a state density process which are
strictly increasing with respect to their initial conditions of which they depend on non-trivial way far from being linear.
The consumption $\bar{\zeta}$ intuitively associated with $\bar{X}$  is given by 
\begin{equation}\label{GlobalConso}
\bar{\zeta}_t(\bar{X}_t(x))=\int_{0}^{+\infty}\zeta^{*,\theta}_t(X^{*,\theta}_t(x^\theta(x)))d\theta=\int_{0}^{+\infty}  x^\theta(x) \zeta^{*,\theta}_t(X^{*,\theta}_t)d\theta 
\end{equation}
where the last equality comes from the linearity, for a fixed $\theta$,   of $\zeta^{*,\theta}$ and $X^{*,\theta}$.
To complete the construction of the progressive utility for which $(\bar{X},\bar{Y},\bar{\zeta})$ will be the  optimal processes, a martingale property on the process  $(  e^{-\int_0^t \bar \zeta_s ds   } \bar{X_t}\bar{Y_t}   )$ is
necessary. We, then, make the following assumption 
\begin{Assumption}\label{hypXYmgle}
 The optimal policies $\kappa^{*,\theta}$ and $\nu^{*,\theta}$ and the market risk premium $\eta$ are uniformly bounded. 
\end{Assumption}

\noindent
Assumption \ref{hypXYmgle} implies that $ ( e^{-\int_0^. \zeta_s^{*, \theta} ds   } X^ {*, \theta} Y^{*, \theta'}) $  are martingales for all ${\theta, \theta'}$ and consequently $(  e^{-\int_0^t \bar \zeta_s ds   } \bar{X_t}\bar{Y_t}   )$ is also  a martingale. \\

\noindent
\underline{{\bf step 3:}} The last step is then to consider any classical utility functions $u$ and $v$ (not necessarily of power type
nor generating linear  optimal processes) and only impose that  their derivatives $ u_x$ and $v_c$ 
have good integrability conditions close to zero. All the ingredients were met, from \eqref{charac} of Theorem \ref{explicite
construction}, by considering the monotonic process $\bar{C}$ defined by
$$\bar{C}(v_c^{-1}(u_x(x))):=\bar{\zeta}_t(\bar{X}_t(x))=\int_{0}^{+\infty}\zeta^{*,\theta}_t(X^{*,\theta}
_t(x^\theta(x)))d\theta,$$
the pair of random fields defined by
\begin{eqnarray}\label{UVcaracterisation}
\left\{
 \begin{array}{cll}
&U(t,x)=\int_0^x\bar{Y}_t(u_x(\bar{\X}(t,z))dz, \\
& V(t,c)=\int_0^c\bar{Y}_t(v_c(\bar{\C}_t(\theta))d\theta \\
%&\text{\rm with} \>\zeta^{*,-1}(0,c)=u_x^{-1}(v_c(c))
 \end{array}
 \right .
\end{eqnarray}
is a consistent progressive utility of investment and consumption generating $(\bar{X},\bar{Y},\bar{C})$ as optimal wealth,
dual and consumption processes, with dual $(\tU,\tV)$:
\begin{eqnarray}\label{UVdualcaracterisation}
\left\{
 \begin{array}{cll}
&\tU(t,y)=\int_y^\infty \bar{X}_t(-\tu_y(\bar{\Y}(t,z))dz, \\
& \tV(t,c)=\int_c^\infty \bar{C}_t(-\tv_c(\bar{\Y}(t,\alpha))d\alpha \\
%&\text{\rm with}\>\zeta^{*}(0,c)=-\tv_c(u_x(c))
 \end{array}
 \right .
\end{eqnarray}
where $(\bar{\X},\bar{\Y},\bar{\C})$ denotes the inverse flows of $\bar{X}$, $\bar{Y}$ and $\bar{C}$.\\
{\bf  Example } Suppose that for any $\theta,\theta'$ we have $X^{*,\theta}=X^{*,\theta'}=X^{*},\forall \theta$ a.s., then in
this case $\bar{X}(x)=xX^{*}$ with inverse $\bar{\X}(x)=x/X^*$, consequently the progressive utility $U$ is given
by:
$$U(t,x)=\int_0^x\int_{0}^{+\infty}y^\theta(u_x(x/X^*_t))Y^{*,\theta}_td\theta dz.$$
%{\bf Financial interpretation:}

\subsection{Application to Ramsey rule evaluation}
Let us now, turn to the Ramsey rule, and study the price of zero coupon. We recall at first that the price in our new
framework  is then given by 
\begin{equation}
 B(t,T)(y)=\mathbb{E}\B[ \frac{\bar{Y}_{T}(y)}{\bar{Y}_t(y)} \b| \mathcal{F}_t  \B]
\end{equation}
From the formula of  $\bar{Y}$ \eqref{BarYMixture}, the price $B(t,T)$ becomes
\begin{equation}
 B(t,T)(y)=\frac{1}{\int_{0}^{+\infty}y^\theta(y)Y^{*,\theta}_td\theta}\mathbb{E}\B[
\int_{0}^{+\infty}y^\theta(y)Y^{*,\theta}_Td\theta \b| \mathcal{F}_t  \B]
\end{equation}
Now,  let us introduce $B^\theta(t,T)$ the zero coupon bond (independent on $y$ because $Y^{*,\theta}$ is linear on $y$)
associated with risk aversion $\theta$ defined by
\begin{equation}
 B^\theta(t,T)=\mathbb{E}\B[ \frac{Y^{*,\theta}_{T}(y)}{Y^{*,\theta}_t(y)} \b| \mathcal{F}_t  \B]
\end{equation}
it follows that 

\begin{equation}
 B(t,T)(y)=\frac{1}{\int_{0}^{+\infty}y^\theta(y)Y^{*,\theta}_td\theta}
\int_{0}^{+\infty}y^\theta(y)Y^{*,\theta}_tB^\theta(t,T)d\theta 
\end{equation}
It is clear from this formula, that the zero coupon  is a mixture  of prices $B^\theta(t,T)$ weighted by $
\frac{y^\theta(y)Y^{*,\theta}_t}{\int_{0}^{+\infty}y^\theta(y)Y^{*,\theta}_td\theta}
$ which is strongly dependent on $ y $ of non-trivial way.\\[1mm]
At this level, several questions naturally arise: What is the sensitivity of the bond with respect to y? It is monotone,
concave, convex? What about its asymptotic behavior?
Give complete and satisfactory answers to these questions is  beyond the scope of this work but will be addressed in a future
paper.
\bibliographystyle{plain}
\bibliography{UtilityConsumption}

\begin{thebibliography}{10}

\bibitem{Jamieson}
David~Jamieson Bolder and Shudan Liu.
\newblock Examining simple joint macroeconomic and term-structure models: A
  practitioner's perspective.
\newblock 2007.

\bibitem{Davis}
Mark~H.A. Davis.
\newblock Option pricing in incomplete markets.
\newblock In S.R. Pliska, editor, {\em Mathematics of Derivative Securities},
  pages 216--226. M.A.H. Dempster and S.R. Pliska, cambridge university press
  edition, 1998.

\bibitem{Duffie}
D.Duffie and L.G.Epstein.
\newblock Stochastic differential utility.
\newblock {\em Econometrica}, 60(2):353--394, 1992.
\newblock With an appendix by the authors and C. Skiadas.

\bibitem{Gollier3}
Christian Gollier.
\newblock What is the socially efficient level of the long-term discount rate?

\bibitem{Gollier16}
Christian Gollier.
\newblock An evaluation of sten's report on the economics of climate change.
\newblock Technical Report 464, IDEI Working Paper, 2006.

\bibitem{Gollier13}
Christian Gollier.
\newblock Comment int{\'e}grer le risque dans le calcul {\'e}conomique?
\newblock {\em Revue d'{\'e}conomie politique}, 117(2):209--223, 2007.

\bibitem{Gollier6}
Christian Gollier.
\newblock The consumption-based determinants of the term structure of discount
  rates.
\newblock {\em Mathematics and Financial Economics}, 1(2):81--101, July 2007.

\bibitem{Gollier14}
Christian Gollier.
\newblock Managing long-term risks.
\newblock 2008.

\bibitem{GollierEcological}
Christian Gollier.
\newblock Ecological discounting.
\newblock IDEI Working Papers 524, Institut d'{\'E}conomie Industrielle (IDEI),
  Toulouse, July 2009.

\bibitem{Gollier15}
Christian Gollier.
\newblock Expected net present value, expected net future value and the ramsey
  rule.
\newblock Technical Report 557, IDEI Working Paper, June 2009.

\bibitem{Gollier9}
Christian Gollier.
\newblock Should we discount the far-distant future at its lowest possible
  rate?
\newblock {\em Economics: the Open Access, Open-Assessment E-Journal},
  3(2009-25), June 2009.

\bibitem{KaratzasShreve:01}
I.Karatzas and S.E.Shreve.
\newblock {\em Methods of Mathematical Finance}.
\newblock Springer, September 2001.

\bibitem{MrNek03}
N.El Karoui and M.~Mrad.
\newblock Mixture of consistent stochastic utilities, and a priori randomness.
\newblock {\em preprint.}, 2010.

\bibitem{MrNek02}
N.El Karoui and M.~Mrad.
\newblock Stochastic utilities with a given optimal portfolio : approach by
  stochastic flows.
\newblock {\em Preprint.}, 2010.

\bibitem{MrNek01}
N.El Karoui and M.~Mrad.
\newblock An exact connection between two solvable sdes and a non linear
  utility stochastic pdes.
\newblock {\em SIAM Journal on Financial Mathematics,}, 4(1):697--736, 2013.

\bibitem{MrNek04}
N.El Karoui, M.~Mrad, and C.~Hillairet.
\newblock Ramsey rule with progressive utility \\ in long term yield curves
  modeling.
\newblock {\em preprint.}, 2014.

\bibitem{Teichmann2}
M.~Keller-Ressel, W.~Schachermayer, and J.~Teichmann.
\newblock Affine processes are regular.
\newblock {\em Probability Theory and Related Fields}, (151):591--611, 2011.

\bibitem{Teichmann1}
M.~Keller-Ressel, W.~Schachermayer, and J.~Teichmann.
\newblock Regularity of affine processes on general state spaces.
\newblock {\em Electronic Journal of Probability}, 43(18):1--17, 2013.

\bibitem{zar-07}
M.Musiela and T.Zariphopoulou.
\newblock Stochastic partial differential equations in portfolio choice.
\newblock {\em Preliminary report}, 2007.

\bibitem{zar-08}
M.Musiela and T.Zariphopoulou.
\newblock Portfolio choice under dynamic investment performance criteria.
\newblock {\em Quantitative Finance}, 9(2):161--170, 2009.

\bibitem{zar-03}
M.~Musiela and T.~Zariphopoulou.
\newblock Backward and forward utilities and the associated pricing systems:
  The case study of the binomial model.
\newblock pages 3--44. Princeton University Press, 2005-2009.

\bibitem{zar-07a}
M.~Musiela and T.~Zariphopoulou.
\newblock Investment and valuation under backward and forward dynamic
  exponential utilities in a stochastic factor model.
\newblock In {\em Advances in mathematical finance}, pages 303--334.
  Birkh\"auser Boston, 2007.

\bibitem{ElKarouiFrachot}
Antoine~Frachot Nicole~ElKaroui and Helyette Geman.
\newblock On the behavior of long zero coupon rates in a no arbitrage
  framework.
\newblock {\em Review of derivatives research}, 1:351--369, 1997.

\bibitem{Dybvig}
J.E.~Ingersol P.H.~Dybvig and S.A. Ross.
\newblock Longfforward and zero-coupon rates can never fall.
\newblock {\em Journal of Business}, 69:1--25, 1996.

\bibitem{Monika}
Monika Piazzesi.
\newblock {\em Handbook of Financial Econometrics}, chapter Affine Term
  Structure Models, pages 691--766.
\newblock 2010.

\bibitem{Ramsey}
F.P. Ramsey.
\newblock A mathematical theory of savings.
\newblock {\em The Economic Journal}, (38):543--559, 1928.

\bibitem{Rogers}
L.C.G. Rogers.
\newblock A mathematical theory of savingsduality in constrained optimal
  investment and consumption problems: A synthesis.
\newblock {\em Working paper, Statistical Laboratory, Cambridge
  University.<http://www.statslab.cam.ac.uk/~chris/>}, 2003.

\bibitem{Weitzman}
Martin~L. Weitzman.
\newblock Why the far-distant future should be discounted at its lowest
  possible rate.
\newblock {\em Journal of Environmental Economics and Management},
  36(3):201--208, November 1998.

\bibitem{Weitzman_review}
Martin~L. Weitzman.
\newblock A review of the the stern review on the economics of climate change.
\newblock {\em Journal of Economic Litterature}, 45:703--724, September 2007.

\end{thebibliography}

\end{document}